\documentclass[a4paper,10pt,twoside]{article}
\usepackage{amsmath,amsfonts,amssymb}
\usepackage{graphicx}
\usepackage{eucal}
\usepackage{mathrsfs}
\usepackage{theorem}
\usepackage{pifont}
\usepackage {epsfig}
\usepackage {graphicx}
\newcommand{\todo}[1]{{\sffamily To do:}}

\newtheorem{theorem}{Theorem}

\newtheorem {lemma}{Lemma}
\newenvironment{proof}{{\flushleft \emph{Proof}:}}{\ding{110}}

\title{Directed Random Market: the equilibrium distribution}

\author{Guy Katriel\\ Department of Mathematics, ORT Braude College,\\ Karmiel, Israel\\}

\date{}

\begin{document}

\maketitle

\begin{abstract} We find the explicit expression for the equilibrium wealth distribution of the Directed Random Market process, recently introduced by Mart\'inez-Mart\'inez and L\'opez-Ruiz \cite{martinez}, which turns out to be a Gamma distribution with shape parameter $\frac{1}{2}$. We also prove the convergence of the
discrete-time process describing the evolution of the distribution of wealth to the equilibrium distribution.
\end{abstract}

\section{Introduction}

In recent years a variety of kinetic exchange models have been investigated within the emerging discipline of Econophysics (see reviews  \cite{ghosh,yakovenko}).
These models involve a population of agents, each possessing a certain amount of wealth, undergoing random pairwise exchanges of wealth according to
certain predefined rules. One is mainly interested in the evolution of the wealth distribution as the process unfolds, and in particular
in its convergence to a limiting distribution and in the characteristics of this equilibrium distribution.

In this work we study the Directed Random Market process, recently introduced by Mart\'inez-Mart\'inez and L\'opez-Ruiz \cite{martinez}.
In this model, when two agents interact, one of them is randomly chosen and transfers a fraction $\epsilon$ of its wealth to the other agent, where
$\epsilon \in [0,1]$ is a uniformly distributed random number. Thus if agents $i,j$ with wealths $m_i,m_j$, interact, and if $j$ is the `winner', then
the agents' wealth following the interaction is given by
\begin{equation}\label{inter}m_i'=(1-\epsilon) m_i,\;\;\;m_j'=m_j+\epsilon m_i.\end{equation}

This process should be contrasted with the well-known Dr\u{a}gulescu - Yakovenko process \cite{drag}, in which the two agents' wealths are pooled and then randomly re-divided, so that
$$m_i'=\epsilon (m_i+m_j),\;\;\;m_j'=(1-\epsilon)(m_i+m_j).$$
Thus in Directed Random Market, exchanges are always uni-directional rather than bi-directional as in the Dr\u{a}gulescu - Yakovenko model.
As we shall see below, this leads to a significant difference in the eventual distribution of wealth at the macro-level.
It should be noted that kinetic exchange models with uni-directional exchanges are already present in the pioneering work of J. Angle \cite{angle,angle1}, but
in Angle's Inequality Process the fraction $\epsilon$ is either fixed for all agents, or a fixed number associated to each agent, rather than a
random number independently drawn in each exchange, as in the Directed Random Market.

Mart\'inez-Mart\'inez and L\'opez-Ruiz \cite{martinez} derived the functional iterations for the discrete-time evolution of the
probability density describing the wealth distribution $p_t(x)$, in the limit of infinitely many agents, so that $p_t(x)dx$ is the fraction of the population whose wealth is in the interval $[x,x+dx]$ at time $t=0,1,2,...$.
It is assumed that
in each time step all agents are randomly paired, and interact according to (\ref{inter}). The functional iterations are given (in slightly different notation than
used in \cite{martinez}) by
\begin{equation}\label{recursion}p_{t+1}(x)=T[p_t](x),\end{equation}
where
\begin{equation}\label{defT}T[p](x)=\frac{1}{2}\int_0^x p_t(x-u) \int_{u}^\infty  \frac{1}{v} p_t(v) dv  du+\frac{1}{2} \int_x^\infty \frac{1}{u} p_t(u)du.\end{equation}
Based on numerical evaluation of these iterations, Mart\'inez-Mart\'inez and L\'opez-Ruiz \cite{martinez} noted that the wealth distribution `piles up' at low values of wealth.

Here we will find an explicit expression for the equilibrium distribution, that is the solutions of $p_w^*(x)$ of the fixed-point problem
$$T[p_w^*]=p_w^*,$$
parameterized by the mean wealth $w$, which is a conserved quantity of the iterations (\ref{recursion}).
\begin{theorem}\label{tip} The equilibrium distributions for the Directed Random Market process are given by the probability density ($w>0$):
\begin{equation}\label{exd}p_w^{*}(x)=\frac{1}{\sqrt{2w \pi x}}e^{-\frac{x}{2w}}.\end{equation}
\end{theorem}
Note that the density (\ref{exd}) corresponds to the Gamma distribution with shape parameter $\frac{1}{2}$ and mean $w$.
This can be compared with the equilibrium distribution for the Dr\u{a}gulescu - Yakovenko process, which is an exponential (Boltzmann-Gibbs) distribution:
\begin{equation}\label{drage}p_w^{**}(x)=\frac{1}{w}e^{-\frac{x}{w}}.\end{equation}
The equilibrium density (\ref{exd}), in contrast with (\ref{drage}), goess to infinity as $x$ goes to zero, which accounts for the `piling up' at low values of wealth noted in \cite{martinez} (see figure \ref{exa3} for plots of the two densities). We note also that the coefficient of variation (standard deviation divided by mean) of
(\ref{drage}) is $CV=1$, while that of (\ref{exd}) is $CV=\sqrt{2}$, which shows that the Directed Random Market leads to a
higher degree of inequality.

\begin{figure}\label{exa3}
\centering
   \includegraphics[height=7cm,width=12cm, angle=0]{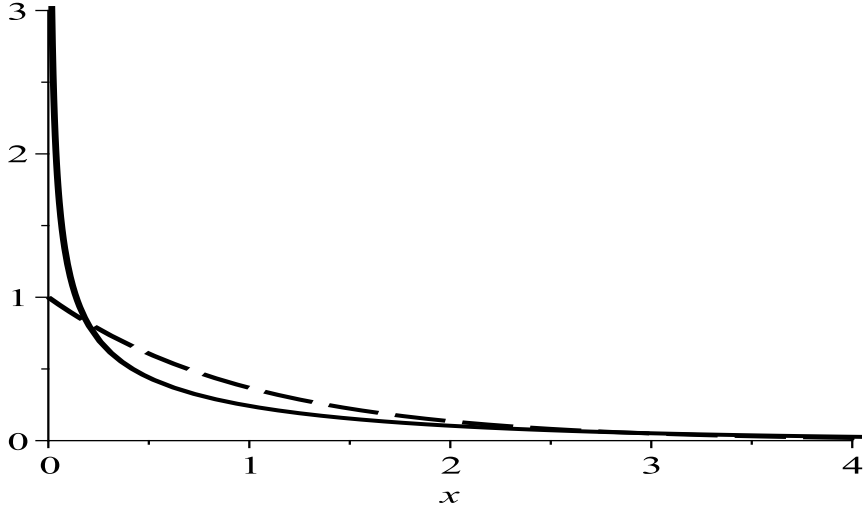}
   \caption{Density of the equilibrium distributions for the Directed Random Market and the Dr\u{a}gulescu - Yakovenko process (dashed line).  In both cases
   the mean wealth is $w=1$.}
\end{figure}

In section \ref{cipp} we prove the convergence of the wealth distribution to the equilibrium distribution, that is,
starting with an arbitrary distribution $p_0(x)$ with mean wealth
\begin{equation}\label{mean}w=\int_0^\infty xp_0(x)dx,\end{equation}
the distributions with densities $p_t$ given by (\ref{recursion}) converge, as $t\rightarrow \infty$, to the distribution with density $p_w^{*}$ given by (\ref{exd}).

\section{Equilibrium distribution for the Directed Random Market}\label{ed}

Equilibrium distributions for the Directed Random Market are fixed points of the operator $T$ given by (\ref{defT}), that is solutions of the functional equation
\begin{equation}\label{eq}p(x)= \frac{1}{2}\int_0^x p(x-u) \int_{u}^\infty  \frac{1}{v} p(v) dv  du+\frac{1}{2} \int_x^\infty \frac{1}{u} p(u)du.\end{equation}

To solve (\ref{eq}) we will use the Laplace transform
$$\hat{p}(s)={\cal{L}}[p](s)=\int_0^\infty e^{-sx}p(x)dx.$$

Using standard properties of the Laplace transform, we have
$${\cal{L}}\Big[\frac{1}{x} p(x)\Big](s)= \int_0^s \hat{p}(s')ds'\;\;\Rightarrow\;\;{\cal{L}}\Big[\int_{x}^\infty  \frac{1}{v} p(v) dv\Big](s) = \frac{1}{s}\int_0^s \hat{p}(s')ds'$$
$$\Rightarrow\;\;{\cal{L}}\Big[\int_0^x p(x-u) \int_{u}^\infty  \frac{1}{v} p(v) dv  du\Big](s)= \hat{p}(s)\cdot\frac{1}{s}\int_0^s \hat{p}(s')ds',$$
$${\cal{L}}\Big[\int_x^\infty \frac{1}{u} p(u)du\Big](s)  = \frac{1}{s}\int_0^s\hat{p}(s')ds',$$
hence
\begin{equation}\label{lll}{\cal{L}}[T[p]](s)=\frac{1}{2s}\cdot[\hat{p}(s)+1 ]\cdot\int_0^s \hat{p}(s')ds',\end{equation}
so that the Laplace-transformed version of (\ref{eq}) is
\begin{equation}\hat{p}(s)=\frac{1}{2s}\cdot[\hat{p}(s)+1 ]\cdot\int_0^s \hat{p}(s')ds'.\end{equation}
This equation can now be readily solved. Isolating the integral and differentiating both sides, we obtain
$$\Big[\frac{s\hat{p}(s)}{\hat{p}(s)+1}\Big]'=\frac{1}{2}\hat{p}(s),$$
that is
$$\hat{p}'(s)=\frac{1}{2s}[(\hat{p}(s))^2-1]\hat{p}(s),$$
a separable differential equation which is solved to yield:
$$\hat{p}(s)=\frac{1}{\sqrt{1+Cs}},$$
and the inverse Laplace transform gives:
$$p(x)=\frac{1}{\sqrt{C\pi x}}e^{-\frac{x}{C}},$$
and determining $C$ by the condition $\int_0^\infty p(x)dx=w$ we obtain $C=2w$, giving the result of Theorem \ref{tip}.

\section{Convergence to the equilibrium distribution}\label{cipp}

In this section we prove that the functional iterations given by (\ref{recursion}), starting
from a general wealth distribution $p_0$, indeed converge
to the equilibrium distribution $p_w^{*}$ given by Theorem \ref{tip}, where $w$ is the initial mean-wealth given by (\ref{mean}).

The proof uses the framework and ideas that we used in \cite{katriel}, following earlier work of L\'opez, L\'opez-Ruiz, and  Calbet \cite{lopez1,lopez2}, where an analogous result was proved for
the iterations arising from the Dr\u{a}gulescu - Yakovenko process. The method employs some key ideas used in the study of related
continuous-time Boltzmann-type equations describing exchange processes, see works of D\"uring, Matthes and Toscani \cite{during,matthes}.

We define by ${\cal{P}}$ the set of all probability densities on $[0,\infty)$, that is the set of a.e. non-negative functions $p\in L^1[0,\infty)$, with $\|p\|_{L^1}=1$.
For any real $\alpha\geq 0 $ the $\alpha$-moment of $p$ is defined as
$$M_\alpha(p)=\int_0^\infty x^\alpha p(x)dx.$$
In particular $M_1(p)$ is the mean wealth corresponding to the density $p$.
For $\alpha\geq 1,w>0$ we define
$${\cal{P}}_{\alpha,w}=\{ p\in {\cal{P}}\;|\; M_\alpha(p)<\infty,\;\;M_1(p)=w\}.$$
Convergence to the equilibrium density $p^{*}_w$ will be proven for initial distributions $p_0 \in {\cal{P}}_{\alpha,w}$ for some $\alpha>1$, although we
conjecture that the result is also true for $p_0 \in {\cal{P}}_{1,w}$.

To each of the probability densities $p_t$ we associate its cumulative probability function
\begin{equation}\label{cum}F_t(x)=\int_0^x p_t(u)du.\end{equation}
We also define the cumulative probability functions associated with the equilibrium densities $p_w^{*}$:
$$F_w^{*}(x)= \int_0^x p^{*}_w(u)du= \frac{1}{\sqrt{2w\pi}}\int_0^x \frac{1}{u}e^{-\frac{u}{2w}}du=\Phi\Big(\sqrt{\frac{x}{2w}} \Big),$$
where $\Phi(x)=\frac{2}{\sqrt{\pi}}\int_0^x e^{-z^2}dz$ is the error function.

Our convergence result is
\begin{theorem}\label{conip} Assume $\alpha>1$, let $p_0\in {\cal{P}}_{\alpha,w}$ be an arbitrary initial wealth distribution, and let the sequence $p_t$ be defined by (\ref{recursion}).
Then the sequence $F_t$ defined by (\ref{cum}) satisfies
$$\lim_{t\rightarrow \infty} F_t(x)=F_w^{*}(x),\;\;\;\forall x\geq 0.$$
\end{theorem}

We begin by showing that the classes ${\cal{P}}_{\alpha,w}$ are invariant under the action of $T$.
\begin{lemma}\label{inv} If $\alpha\geq 1$ and $p\in {\cal{P}}_{\alpha,w}$ then $T[p]\in {\cal{P}}_{\alpha,w}$.
\end{lemma}
\begin{proof}
We first need to prove the finiteness of
\begin{eqnarray}\label{i1}&&M_\alpha(T[p])=\int_0^\infty x^\alpha T[p](x)dx\\
& =& \frac{1}{2}\int_0^\infty x^\alpha \int_0^x p(x-u) \int_{u}^\infty  \frac{1}{v} p(v) dv  du dx+\frac{1}{2} \int_0^\infty x^\alpha\int_x^\infty \frac{1}{u} p(u)du dx.\nonumber\end{eqnarray}
By changing order of integration and using the inequality $(x+u)^\alpha\leq 2^{\alpha-1}(x^\alpha+u^\alpha)$ (for $\alpha\geq 1$) we obtain
\begin{eqnarray}\label{i2}&&\int_0^\infty x^\alpha \int_0^x p(x-u) \int_{u}^\infty  \frac{1}{v} p(v) dv  du dx\nonumber\\
&=&\int_0^\infty  \int_u^\infty x^\alpha p(x-u)dx \int_{u}^\infty  \frac{1}{v} p(v) dv   du\nonumber\\
&=&\int_0^\infty  \int_0^\infty (x+u)^\alpha p(x)dx \int_{u}^\infty  \frac{1}{v} p(v) dv   du\nonumber\\
&\leq& 2^{\alpha-1}\int_0^\infty  \int_0^\infty x^\alpha p(x)dx \int_{u}^\infty  \frac{1}{v} p(v) dv   du+\int_0^\infty  u^\alpha\int_0^\infty  p(x)dx \int_{u}^\infty  \frac{1}{v} p(v) dv   du\nonumber\\
&=& 2^{\alpha-1}M_\alpha(p) \int_0^\infty   \int_{u}^\infty  \frac{1}{v} p(v) dv   du+\int_0^\infty  u^\alpha \int_{u}^\infty  \frac{1}{v} p(v) dv   du\nonumber\\
&=& 2^{\alpha-1}M_\alpha(p) \int_0^\infty \frac{1}{v}p(v) \int_0^v du dv+\int_0^\infty  \frac{1}{v} p(v) \int_{0}^v u^\alpha  du dv \nonumber\\
&=& 2^{\alpha-1}M_\alpha(p) \int_0^\infty p(v)  dv+\frac{1}{\alpha+1}\int_0^\infty  v^{\alpha} p(v)  dv=2^{\alpha-1}\cdot \frac{\alpha+2}{\alpha+1}\cdot M_\alpha(p). \end{eqnarray}
We also have
\begin{eqnarray}\label{i3}&&\int_0^\infty x^\alpha\int_x^\infty \frac{1}{u} p(u)du dx=\int_0^\infty \frac{1}{u} p(u)\int_0^u x^\alpha  dx du \nonumber\\ &=&\frac{1}{\alpha+1}\int_0^\infty u^\alpha p(u)du
=\frac{1}{\alpha+1}M_\alpha(p).\end{eqnarray}
Combining (\ref{i1}),(\ref{i2}) and (\ref{i3}) we conclude that
\begin{equation}\label{bnd1}M_{\alpha}(T[p])\leq \frac{1}{2(\alpha+1)}\cdot\Big(2^{\alpha-1}\cdot (\alpha+2) +1\Big)\cdot M_\alpha(p),\end{equation}
and in particular we have $M_{\alpha}(T[p])<\infty$.

Taking $\alpha=1$, the inequality in (\ref{i2}), and hence the one in (\ref{bnd1}), is in fact an equality, so that we get $M_1(T[p])=M_1(p)=w$, and
we have shown $T[p]\in {\cal{P}}_{\alpha,w}$.
\end{proof}

We now define the following metric on the set ${\cal{P}}_{\alpha,w}$, where we now assume $\alpha\in (1,2)$.
$$p,q\in {\cal{P}}_{\alpha,w}\;\;\;\Rightarrow\;\;\; d_{\alpha}(p,q)=\sup_{s>0}\frac{|{\cal{L}}[p](s)-{\cal{L}}[q](s)|}{s^\alpha}.$$

The finiteness of $d_{\alpha,w}$ is ensured by the following Lemma (we refer to {\cite{katriel}}, Lemma 2.3, for the proof):
\begin{lemma}\label{fin} If $1<\alpha<2$, $w>0$, $p,q\in {\cal{P}}_{\alpha,w}$, then $d_{\alpha}(p,q)<\infty$.
\end{lemma}

We now prove that the map $T:{\cal{P}}_{\alpha,w}\rightarrow {\cal{P}}_{\alpha,w}$ is contracting.
\begin{lemma}\label{lip} If $1< \alpha<2$, $w>0$, $p,q\in {\cal{P}}_{\alpha,w}$, then
$$d_{\alpha}(T[p],T[q])\leq \Big(\frac{1}{2}+\frac{1}{\alpha+1}\Big) d_\alpha(p,q).$$
\end{lemma}

\begin{proof} We set $\hat{p}(s)={\cal{L}}[p](s),\hat{q}(s)={\cal{L}}[q](s)$.
Making a change of variable $s'=s\xi$ in the integral in (\ref{lll}) we have
$${\cal{L}}[T[p]](s)=\frac{1}{2}[\hat{p}(s)+1 ]\int_0^1 \hat{p}(s\xi)d\xi.$$
Therefore
$${\cal{L}}[T[p]](s)-{\cal{L}}[T[q]](s) =\frac{1}{2}[\hat{p}(s)-\hat{q}(s) ]\int_0^1 \hat{p}(s\xi)d\xi+\frac{1}{2}[\hat{q}(s)+1 ]\int_0^1 [\hat{p}(s\xi)-\hat{q}(s\xi)]d\xi$$
and using the fact that, since $\|p\|_{L^1}=\|q\|_{L^1}=1$ we have $|\hat{p}(s)|\leq 1,|\hat{q}(s)|\leq 1$, we can estimate
$$\frac{|{\cal{L}}[T[p]](s)-{\cal{L}}[T[q]](s)|}{s^\alpha} \leq \frac{1}{2}\frac{|\hat{p}(s)-\hat{q}(s)|}{s^\alpha}\int_0^1 \hat{p}(s\xi)d\xi$$
$$+\frac{1}{2}[\hat{q}(s)+1 ]\int_0^1 \xi^\alpha \frac{|\hat{p}(s\xi)-\hat{q}(s\xi)|}{(s\xi)^\alpha} d\xi$$
$$\leq \frac{1}{2}d_\alpha(p,q)+d_\alpha(p,q)\int_0^1 \xi^\alpha  d\xi=\Big(\frac{1}{2}+\frac{1}{\alpha+1}\Big)d_{\alpha}(p,q),$$
and taking the supremum over $s>0$ we obtain the result.
\end{proof}

We therefore have
\begin{lemma}\label{cm}
If $1< \alpha<2$, $w>0$, $p_0\in {\cal{P}}_{\alpha,w}$, and the sequence $\{p_t\}_{t=0}^\infty$ is defined by (\ref{recursion}),  then
\begin{equation}\label{limi} \lim_{t\rightarrow \infty} d_\alpha(p_t,p_w^{*})=0. \end{equation}
\end{lemma}
\begin{proof}
We show that
\begin{equation}\label{inn}d_\alpha(p_t,p_w^{*})\leq d_\alpha(p_0,p_w^{*})\Big(\frac{1}{2}+\frac{1}{\alpha+1}\Big)^t,\end{equation}
which implies the result since $\alpha>1$ implies $\frac{1}{2}+\frac{1}{\alpha+1}<1$. For $t=0$, (\ref{inn}) is trivial. We proceed by induction, using
Lemma \ref{lip} and the fact that $p_w^{*}$ is a fixed point of $T$:
$$d_\alpha(p_{t+1},p_w^{*})=d_\alpha(T[p_t],T[p_w^{*}])\leq \Big(\frac{1}{2}+\frac{1}{\alpha+1}\Big)d_\alpha(p_t,p_w^{*})$$
$$\leq d_\alpha(p_0,p_w^{*}) \Big(\frac{1}{2}+\frac{1}{\alpha+1}\Big)^{t+1}.$$
\end{proof}

Convergence in the metric $d_{\alpha}$ implies pointwise convergence of the Laplace transform,
\begin{equation}\label{pcon}\lim_{t\rightarrow \infty}{\cal{L}}[p_t](s)= {\cal{L}}[p_w^{*}](s),\;\;\;s>0. \end{equation}
We note that the assumption $\alpha<2$ is irrelevant for the validity of (\ref{pcon}), since ${\cal{P}}_{\alpha,w}\subset {\cal{P}}_{\beta,w}$ when $\alpha>\beta$,
so we have (\ref{pcon}) whenever $p_0\in {\cal{P}}_{\alpha,w}$  for some $\alpha>1$.

Mukherjea et al. \cite{muk} proved that pointwise convergence of a sequence of Laplace transforms $\hat{p}_t(s)$ of probability densities $p_t$ to the Laplace transform $\hat{p}(s)$ of a probability density $p$, for all $s$ in some interval, implies the pointwise convergence of the corresponding cumulative probability
functions. Therefore the pointwise convergence of $\hat{p}_t$ to $\hat{p}_w^*$, which we have established, implies Theorem \ref{conip}.

\end{document}